\documentclass[10pt]{article}
\usepackage[margin=1in]{geometry}
\usepackage{amsmath,amssymb,amsthm,mathtools}
\usepackage{enumitem}
\usepackage[hidelinks]{hyperref}
\usepackage{microtype}
\setlist[itemize]{leftmargin=1.2em,itemsep=1pt,topsep=2pt}
\setlist[enumerate]{leftmargin=1.4em,itemsep=1pt,topsep=2pt}

\newtheorem{theorem}{Theorem}
\newtheorem{proposition}{Proposition}
\newtheorem{lemma}{Lemma}
\newtheorem{definition}{Definition}
\newtheorem{remark}{Remark}

\newcommand{\rhohist}{\rho_t^{\mathrm{hist}}}
\newcommand{\rhomin}{\rho_t^{\min}}
\newcommand{\rhonew}{\rho_t^{\mathrm{new}}}
\newcommand{\rhonewmin}{\rho_t^{\mathrm{new,min}}}
\newcommand{\rhoavail}{\rho_t^{\mathrm{avail}}}

\newcommand{\lb}{\mathrm{LB}}

\newcommand{\Cans}{C_{\mathrm{ans}}}
\newcommand{\Nsat}{N_{\mathrm{sat}}}

\title{Residual Privacy Budgeting with Weighted Scarcity Allocation for Online Query Answering}
\author{Mina Khoshmehr \\ The University of Auckland \and Fernando Beltran \\ The University of Auckland}
\date{}

\begin{document}
\maketitle

\begin{abstract}
In many practical deployments of differential privacy, queries do not arrive all at once. We study online differentially private query answering under a finite zero-concentrated differential privacy (zCDP) contract. In this setting, queries arrive sequentially, carry different accuracy thresholds, and may overlap with information already released. We formulate this setting as residual privacy budgeting: for each arriving query, the mechanism first credits reusable support from previous DP outputs and then spends new budget only on the remaining support required to satisfy the current threshold. The controller separates feasible cases, where the minimal residual support is allocated exactly, from scarcity cases, where a weighted shortfall-conservation optimiser assigns limited support according to query difficulty. We define the weight using the Query Influence Factor (QIF), a diagnostic signal for query difficulty and instability rather than query importance. For scalar Gaussian exact reuse, inverse-variance fusion justifies additive support. We prove zCDP composition, residual minimality, 1-competitiveness against the offline optimum in the feasible regime, and avoidable expenditure for allocators that ignore released history. A scarcity impossibility result shows that no online allocator can guarantee a competitive ratio better than \(1/n\) in threshold satisfaction, contextualising the QIF scarcity layer as a design choice for an inherently hard online problem.
\end{abstract}

\section{Introduction}
Differential privacy (DP) has evolved from a theoretical construct into a practical requirement for modern data systems. The primary challenge in practice is not how to answer a single query privately, but how to allocate a limited global privacy budget across a sequence of queries whose future features are unknown. This turns privacy into an online resource allocation problem.

Workload-aware mechanisms improve the utility of DP by exploiting the structure in a known workload. The matrix mechanism, HDMM, histogram-consistency methods, and recent scalable matrix-mechanism variants for noisy marginals optimise measurement strategies, consistency constraints, or covariance-aware plans for fixed or batch workloads \cite{li2015,hay2010,mckenna2023hdmm,xiao2023residualplanner}. Iterative mechanisms such as MWEM and private multiplicative weights use transcript-aware measurement and update rules, but their primary object is measurement selection or distributional approximation rather than residual support allocation under per-query thresholded utility \cite{hardt2010,hardt2012}. Accuracy-aware systems such as APEx price queries against explicit utility targets, but do not explicitly credit support already present in previously released DP outputs \cite{ge2019}. Systems work on privacy-budget management, including PINQ, Sage, PrivateKube, and Cohere, treats privacy as an operational resource to be tracked, scheduled, or allocated across applications \cite{mcsherry2009pinq,lecuyer2019sage,luo2021privatekube,kuchler2024cohere}. Recent work on online query answering with multiple analysts also studies budget sharing under online query arrivals, but with a focus on fairness and ordering effects across analysts rather than residual support reuse \cite{pujol2022multi}.

Our setting is an instance of online resource allocation under uncertainty. The classical competitive analysis framework \cite{borodin1998online} measures an online algorithm's worst-case performance ratio against an offline optimum. In the online knapsack problem \cite{marchetti1995stochastic}, items of varying size and value arrive sequentially and must be irrevocably accepted or rejected into a capacity-limited knapsack. Our setting differs in two ways: (1) support is additive, so partial allocation is possible and multiple allocations to the same query accumulate via fusion; and (2) queries can overlap with previously released outputs, enabling reuse. We show that these features make the feasible regime strictly easier than general online knapsack, where constant competitive ratios require randomisation or distributional assumptions \cite{zhou2008budget}. In the scarcity regime, however, the problem retains the hardness of online selection.

Our mechanism differs from prior DP work by separating two questions: how much support is still needed after crediting released history, and how scarce remaining budget should be distributed when not all residual targets can be met. The first question is handled by residual allocation, while the second is handled by a difficulty-weighted scarcity layer. We define this weight using the Query Influence Factor (QIF), which is not an importance ranking over queries, but a diagnostic signal for matching privacy support to query difficulty and instability. We do not claim superiority over transcript-aware online matrix, HDMM-style, MWEM-style, or multi-analyst mechanisms. Extending the framework to those settings would require replacing the scalar support term \(\rho_t^{\mathrm{hist}}\) with a workload-specific representation of how previous noisy releases reduce the uncertainty of the current query.

\section{Problem Statement}
In practical differential privacy (DP) deployments, demand is often sequential: queries arrive at different times, require different levels of accuracy, and may overlap with information already released. Static workload-aware approaches optimise a known workload in advance. By contrast, in our setting the curator observes only the current query and the public transcript so far; the future workload is unknown at decision time.
\paragraph{Setting assumptions.}
We study a trusted-curator, online allocation setting with a fixed zCDP contract.
The query sequence is non-adaptive: queries may arrive sequentially and may overlap with earlier releases, but their generation does not depend on the realised randomness of previous DP outputs. The allocator observes only the current query, declared accuracy requirement, remaining budget, and the public transcript of previous DP releases. It does not optimise a known workload in advance. In the main formal results, historical reuse is triggered only by exact public-descriptor equality, and scalar support additivity is proved only for Gaussian exact reuse. Extensions to correlated linear workloads require a non-scalar support representation and are left outside the present scope.

Let $D$ be a sensitive dataset and let $q_1,\ldots,q_K$ be a non-adaptive sequence of bounded-sensitivity queries. At step $t$, query $q_t$ arrives with accuracy threshold $U_t$ and confidence level $1-\gamma_t$. Here \(K\) denotes the finite analysis horizon used to state composition and comparison results. It is not a separate privacy parameter; privacy loss is controlled by the zCDP contract \(\rho_{\mathrm{tot}}\). The curator operates under a finite zCDP contract $\rho_{\mathrm{tot}}=\rho_{\mathrm{est}}+\rho_{\mathrm{ans}}$: diagnostics consume estimation budget and new answers consume answering budget. The public transcript after $t-1$ releases is
\begin{equation}
H_{t-1}=\{(q_j,\tilde q_j,\rho_j)\}_{j<t},
\end{equation}
where $\tilde q_j$ is a DP answer and $\rho_j$ is its accounted support under the selected QueryModel. Each query has a monotone error surrogate $r_t(\rho)$, non-increasing in effective support $\rho$. The from-scratch support required by the declared threshold is
\begin{equation}
\rhomin=\inf\{\rho\geq 0:r_t(\rho)\leq U_t\}.
\end{equation}
The value $\rhomin$ is random before Phase I, because it is computed from a privatised diagnostic; all feasibility and uniqueness claims are conditional on the realised diagnostic transcript. Historical support $\rhohist\geq0$ is computed from public query descriptors and prior DP outputs. In the main formal results, reuse is triggered by exact public-descriptor equality, so $\rhohist$ is the precision-equivalent support of an exact historical Gaussian release. Approximate reuse with an external or public-domain bias envelope is appendix-only. The residual need is
\begin{equation}
\label{eq:residual}
\rhonewmin=\max\{0,\rhomin-\rhohist\}.
\end{equation}

In this paper, the support-additivity argument is proved for the scalar Gaussian exact-reuse case. This means that the current query exactly matches a previously released query, and both the historical and new answers are Gaussian estimates of the same scalar quantity. In that setting, inverse-variance fusion gives effective support
\begin{equation}
\rho_t^{\mathrm{eff}}
=
\rho_t^{\mathrm{hist}}
+
\rho_t^{\mathrm{new}},
\end{equation}
which justifies the residual rule in~\eqref{eq:residual}. For more complex workloads, such as correlated histograms, range queries, or general linear queries, historical support may not be representable by a single scalar \(\rho_t^{\mathrm{hist}}\). Extending the controller to those workloads would require specifying how previous noisy releases reduce uncertainty for the current query, through an appropriate support representation, diagnostic rule, and fusion rule.
Approximate reuse is treated separately in Appendix~\ref{app:bias}. There, a public or externally supplied bias envelope \(B_t\) controls the mismatch between a historical query and the current query, yielding a bias--variance MSE bound.
The main theorems in the paper do not rely on approximate reuse.

\paragraph{Contributions.}
First, we formulate online DP query answering as residual support allocation under thresholded utility, where new privacy budget is charged only for support not already provided by released DP outputs. Second, we state an explicit zCDP accounting boundary for the public transcript: diagnostics and new answers consume privacy budget, while residual allocation and fusion are post-processing of previous DP outputs and public state. Third, we prove residual minimality in the feasible regime, give an explicit scarcity rule showing how any positive difficulty weight affects allocation under shortage, and prove that the residual allocator is 1-competitive against the offline optimum in answering expenditure. Fourth, we define QIF as an interpretable diagnostic signal for query difficulty and instability, contextualised by a scarcity impossibility result showing that no online allocator can guarantee a competitive ratio better than \(1/n\) in threshold satisfaction. Fifth, we provide formal separations from accuracy-aware non-crediting allocators \cite{ge2019} and budget-partitioning allocators \cite{pujol2022multi}, together with an exact-repeat witness separation and a diagnostic-amortisation identity.

\section{Mechanism}
For each arriving query, the controller proceeds through four phases. \emph{Lookup:} search $H_{t-1}$ using public query descriptors; exact matches are used in the main results. \emph{Diagnostics:} release a diagnostic signal at cost $\rho_{\mathrm{probe},t}$ and compute $\rhomin$ by a conservative lower-confidence construction. \emph{Residual allocation:} compute~\eqref{eq:residual}; if $\rhonewmin\leq\rhoavail$, allocate $\rhonew=\rhonewmin$, otherwise solve the scarcity problem in Section~\ref{sec:scarcity}. \emph{Execution and fusion:} if $\rhonew>0$, release a new answer at support $\rhonew$ and fuse it with reusable historical outputs.
Fusion is the post-processing step that combines reusable historical DP outputs with any newly released DP answer for the current query. In the scalar Gaussian exact-reuse case, the historical and new answers estimate the same scalar query, so inverse-variance fusion gives query-level effective support
\begin{equation}
\rho_t^{\mathrm{eff}}
=
\rho_t^{\mathrm{hist}}
+
\rho_t^{\mathrm{new}}.
\end{equation}
This effective support is used for accuracy, while only \(\rho_t^{\mathrm{new}}\) is newly charged to the answering budget.

The decision logic is therefore simple. The controller first checks whether released history already supports the query; if so, no new answering budget is spent. If additional support is needed, the controller computes the residual requirement. QIF is consulted only when this residual exceeds the available budget. Thus residual need decides whether and how much new budget is required, while QIF affects only the scarcity allocation.

We do not claim that post-processing itself is new. Post-processing is the accounting boundary that makes the public allocation transcript safe to publish.

\begin{theorem}[Public allocation transcript]
\label{thm:postprocessing}
Let $\widetilde S_t$ denote the DP outputs available at step $t$ and let $s_t$ denote public state. If $\rhonew=f_t(\widetilde S_t,s_t)$, where $f_t$ is deterministic or randomised independently of $D$, then publishing $\rhonew$ incurs no privacy loss beyond that already paid for $\widetilde S_t$.
\end{theorem}
\begin{proof}[Proof sketch]
The allocation is a measurable function of privatised outputs and public state. By post-processing closure of zCDP, such a map cannot increase privacy loss.
\end{proof}

\begin{theorem}[End-to-end zCDP]
\label{thm:zcdp}
If $\sum_t\rho_{\mathrm{probe},t}\leq \rho_{\mathrm{est}}$ and $\sum_t\rhonew\leq \rho_{\mathrm{ans}}$, then the full mechanism satisfies $\rho_{\mathrm{tot}}$-zCDP.
\end{theorem}
\begin{proof}[Proof sketch]
Only diagnostics and new query execution access $D$. Lookup, residual computation, allocation, and fusion are post-processing by Theorem~\ref{thm:postprocessing}. Additive zCDP composition gives the claim.
\end{proof}

\section{Residual and Weighted Scarcity Allocation}
\label{sec:scarcity}
The controller separates residual need from difficulty-sensitive scarcity allocation. Residual need determines whether new budget is required; difficulty weights determine how scarce support is matched to queries only when the residual cannot be fully supported.

\begin{theorem}[Residual minimality]
\label{thm:minimality}
Each query has an error surrogate \(r_t(\rho)\), which maps effective support \(\rho\) to the certified error bound for query \(q_t\). We assume \(r_t(\rho)\) is strictly decreasing on the relevant support interval, so larger support corresponds to lower error. Conditional on the realised diagnostic transcript, if $\rhonewmin\leq \rhoavail$, then $\rhonew=\rhonewmin$ is the unique minimal new allocation satisfying the threshold for $q_t$.
\end{theorem}
\begin{proof}[Proof sketch]
If $\rhohist\geq\rhomin$, no new support is needed. Otherwise any allocation below $\rhomin-\rhohist$ leaves total support below $\rhomin$, and strict monotonicity of $r_t$ cannot certify the target.
\end{proof}

\begin{remark}[QIF is inactive in the feasible regime]
When $\rhonewmin\leq\rhoavail$, the allocation is $\rhonew=\rhonewmin$, independently of the QIF weight. QIF cannot inflate a feasible query as if it were unsupported. It affects allocation only in the scarcity regime, where $\rhonewmin>\rhoavail$.
\end{remark}

In scarcity, $\rhonewmin>\rhoavail$. For any positive difficulty weight $w_t>0$, the allocator solves
\begin{equation}
\label{eq:scarcity}
\min_{0\leq\rho\leq\rhoavail}w_t[\rhomin-(\rhohist+\rho)]_+^2+\mu_t\rho,
\end{equation}
where $\mu_t>0$ is a conservation parameter that penalises current expenditure under scarcity.

This proposition is not intended as a new optimisation result; its role is to make the controller's scarcity behaviour explicit and to show how the difficulty weight enters the allocation rule.
\begin{proposition}[Closed-form scarcity rule]
\label{thm:scarcity}
Let \(a_t=\rho_t^{\min}-\rho_t^{\mathrm{hist}}\). In the active scarcity regime
\(a_t>\rho_t^{\mathrm{avail}}\), the weighted shortfall-conservation objective
has the explicit solution
\[
\rho_t^\star
=
\Pi_{[0,\rho_t^{\mathrm{avail}}]}
\left(
a_t-\frac{\mu_t}{2w_t}
\right).
\]
\end{proposition}
\begin{proof}[Proof sketch]
In active scarcity, the hinge is active on the feasible interval, so the objective is the strictly convex quadratic \(w_t(a_t-\rho)^2+\mu_t\rho\). Differentiating gives the unconstrained minimiser \(a_t-\mu_t/(2w_t)\), and projection gives the constrained solution.
\end{proof}

\section{QIF-weighted scarcity differentiation}

The framework's correctness does not depend on the specific form of the difficulty weight beyond positivity. We define the weight using the \emph{Query Influence Factor} (QIF), not in the information-theoretic sense of quantitative information flow.

\begin{definition}[Query Influence Factor]
\label{def:qif}
For each active-scarcity query \(i\), let \(m_i^{\mathrm{LB}}>0\) be a private lower-confidence scale surrogate, \(\tilde v_i\geq 0\) a diagnostic variability or instability signal, and \(b_i\geq 0\) a rarity or imbalance signal. Define
\begin{equation}
\label{eq:qif}
\operatorname{QIF}_i
=
\frac{\tilde v_i+\beta b_i}{g(m_i^{\mathrm{LB}})},
\qquad
w_i=\operatorname{QIF}_i^\alpha,
\end{equation}
where \(\alpha>0\), \(\beta\geq 0\), and \(g:\mathbb{R}_{>0}\to\mathbb{R}_{>0}\) is continuous and increasing. If \(\operatorname{QIF}_i=0\), the implementation may use an arbitrarily small positive floor so that \(w_i>0\).
\end{definition}

QIF is not an importance ranking over queries. It is a diagnostic difficulty weight used to match limited privacy support to the intrinsic difficulty and instability of each query.

\begin{theorem}[QIF-weighted scarcity differentiation]
\label{thm:qif-differentiation}
Consider two active-scarcity queries \(s\) and \(t\), with residual gaps \(a_i=\rho_i^{\min}-\rho_i^{\mathrm{hist}}\). Suppose both queries use the scarcity objective~\eqref{eq:qif-scarcity-objective}
with common \(\mu>0\), and suppose projection does not clip either optimiser. Then the QIF weight is non-decreasing in \(\tilde v_i\) and \(b_i\), and non-increasing in \(m_i^{\mathrm{LB}}\). Moreover,
\begin{equation}
\label{eq:qif-compact-result}
\rho_i^\star=a_i-\frac{\mu}{2w_i},
\qquad
S_i^\star=a_i-\rho_i^\star=\frac{\mu}{2w_i}.
\end{equation}
If \(a_s=a_t\) and \(w_t>w_s\), then
\begin{equation}
\label{eq:qif-compact-difference}
\rho_t^\star-\rho_s^\star=
\frac{\mu}{2}\left(\frac{1}{w_s}-\frac{1}{w_t}\right)>0,
\qquad
S_t^\star<S_s^\star.
\end{equation}
\end{theorem}

The scarcity objective referenced above is:
\begin{equation}
\label{eq:qif-scarcity-objective}
\min_{0\leq \rho_i\leq \rho_i^{\mathrm{avail}}}
w_i(a_i-\rho_i)^2+\mu\rho_i.
\end{equation}

\paragraph{Proof.} See Appendix~\ref{app:proofs}.

\section{Comparisons, Competitive Analysis, and Limits}
\label{sec:comparisons}
We establish three types of comparison: structural separations from allocators that do not credit released history, a competitive ratio analysis against the offline optimum, and a scarcity impossibility result.

\subsection{Separation from non-crediting allocators}
\begin{definition}[State-agnostic threshold-pricing]
\label{def:sa}
An allocator $\pi^{\mathrm{SA}}$ may depend on $(U_t,\gamma_t,\rhoavail)$ and the current privatised diagnostic, but not on $\rhohist$. It is threshold-pricing if, whenever the from-scratch target is feasible, it allocates enough support to satisfy the threshold from scratch.
\end{definition}

\begin{definition}[Accuracy-aware non-crediting allocator]
\label{def:ac}
An allocator \(\pi^{\mathrm{AC}}\) is \emph{accuracy-aware non-crediting} if, at each step \(t\), it computes a from-scratch support allocation that depends on \((q_t, U_t, \gamma_t, \rho_t^{\mathrm{avail}})\) and possibly a privatised diagnostic, but is functionally independent of \(\rho_t^{\mathrm{hist}}\). Unlike the purely state-agnostic class (Definition~\ref{def:sa}), \(\pi^{\mathrm{AC}}\) may use query-specific diagnostic information, but it does not fuse or credit previously released DP outputs.
\end{definition}

This definition captures the essential structural property of APEx's pricing step \cite{ge2019}: queries are costed against their accuracy target, but no mechanism exists to reduce that cost based on what has already been released. The definition does not claim to capture every implementation detail of APEx (e.g., its SVT-based exploration strategy), only the property relevant to the separation.

\begin{definition}[Budget-partitioning allocator]
\label{def:bp}
An allocator \(\pi^{\mathrm{BP}}\) divides \(\rho_{\mathrm{ans}}\) among \(A\) analysts, giving analyst \(a\) a sub-budget \(\rho^{(a)}\) with \(\sum_a \rho^{(a)}\leq\rho_{\mathrm{ans}}\). Within each analyst's sub-budget, queries are answered from scratch without crediting the public transcript \(H_{t-1}\).
\end{definition}

This abstracts the structural property of Pujol et al.\ \cite{pujol2022multi} relevant to the comparison: budget is partitioned among competing analysts, but per-query expenditure is not reduced by fusing with previously released outputs. Pujol et al.\ address a different design question---how to partition a shared budget fairly---rather than how to reduce per-query expenditure through reuse. The two questions are complementary, and combining multi-analyst fairness with residual crediting is an interesting open direction.

\begin{theorem}[Avoidable expenditure under non-crediting allocators]
\label{thm:avoidable-expenditure}
Fix a horizon \(K\), an exact-overlap fraction
\(\alpha_{\mathrm{ov}}\in[0,1]\), and average target-implied support
\(\bar{\rho}^{\min}\). There exists a non-adaptive bounded-sensitivity
workload with empirical exact-overlap \(\alpha_{\mathrm{ov}}\) on which any
allocator \(\pi^{\mathrm{SA}}\) that does not credit historical support, that is, whose new-budget map is independent of \(\rho_t^{\mathrm{hist}}\), incurs
new answering expenditure at least
\begin{equation}
\sum_{t=1}^{K} \rho^{\mathrm{new}}_{t,\pi^{\mathrm{SA}}}
\geq
\sum_{t=1}^{K} \rho^{\mathrm{new}}_{t,\pi^\star}
+
\Omega\!\left(\alpha_{\mathrm{ov}}K\bar{\rho}^{\min}\right).
\end{equation}
The same bound holds for any accuracy-aware non-crediting allocator \(\pi^{\mathrm{AC}}\) (Definition~\ref{def:ac}) and for any budget-partitioning allocator \(\pi^{\mathrm{BP}}\) (Definition~\ref{def:bp}) when cross-analyst query overlap is present, since both classes are functionally independent of \(\rho_t^{\mathrm{hist}}\).
\end{theorem}

\begin{proposition}[Per-query budget separation from APEx-style allocators]
\label{prop:apex-separation}
Let \(\mathcal{R}=\{t : q_t=q_j \text{ for some } j<t,\;\rho_t^{\mathrm{hist}}\geq\rho_t^{\min}\}\) be the set of fully supported exact repeats. Then
\[
\sum_{t=1}^{K}\rho_{t,\pi^{\mathrm{AC}}}^{\mathrm{new}}
-
\sum_{t=1}^{K}\rho_{t,\pi^\star}^{\mathrm{new}}
\geq
\sum_{t\in\mathcal{R}}\rho_t^{\min}.
\]
At each \(t\in\mathcal{R}\), the residual allocator spends zero while \(\pi^{\mathrm{AC}}\) pays full from-scratch cost.
\end{proposition}
\begin{proof}[Proof sketch]
At each \(t\in\mathcal{R}\), \(\rho_t^{\mathrm{hist}}\geq\rho_t^{\min}\) so the residual allocator sets \(\rho_{t,\pi^\star}^{\mathrm{new}}=0\). Since \(\pi^{\mathrm{AC}}\) is independent of \(\rho_t^{\mathrm{hist}}\), it allocates at least \(\rho_t^{\min}\) from scratch. Summing over \(\mathcal{R}\) gives the bound.
\end{proof}

\begin{proposition}[Cross-analyst reuse gap]
\label{prop:pujol-separation}
Consider a single-dataset setting with \(A\geq 2\) analysts. Suppose analyst \(a_1\) submits query \(q\) at time \(t_1\), answered at support \(\rho_1\). Later, analyst \(a_2\) submits the same query \(q\) at time \(t_2>t_1\) with accuracy threshold requiring \(\rho^{\min}\leq\rho_1\). Under any budget-partitioning allocator \(\pi^{\mathrm{BP}}\),
\(\rho_{t_2,\pi^{\mathrm{BP}}}^{\mathrm{new}}\geq\rho^{\min}\),
while under the residual allocator
\(\rho_{t_2,\pi^\star}^{\mathrm{new}}=0\).
\end{proposition}
\begin{proof}[Proof sketch]
Under \(\pi^{\mathrm{BP}}\), analyst \(a_2\)'s query is priced from scratch within \(a_2\)'s sub-budget. Under \(\pi^\star\), the public transcript contains \((q,\tilde q_{t_1},\rho_1)\). Since \(\rho_{t_2}^{\mathrm{hist}}=\rho_1\geq\rho^{\min}\), the residual is zero. The crediting step is post-processing by Theorem~\ref{thm:postprocessing}, so no additional privacy is consumed.
\end{proof}

\begin{theorem}[Per-query comparison]
\label{thm:domination}
Let $\pi^\star$ be the residual allocator and let $\pi^{\mathrm{SA}}$ be state-agnostic threshold-pricing. Whenever the from-scratch threshold is feasible for $\pi^{\mathrm{SA}}$,
$
\rho^{\mathrm{new}}_{t,\pi^\star}\leq \rho^{\mathrm{new}}_{t,\pi^{\mathrm{SA}}}
$,
with strict inequality in the feasible exact-reuse regime whenever $\rhohist>0$. Moreover,
\begin{equation}
\label{eq:ratio}
\frac{\rho^{\mathrm{new}}_{t,\pi^\star}}{\rho^{\mathrm{new}}_{t,\pi^{\mathrm{SA}}}}
\leq\max\left\{0,1-\frac{\rhohist}{\rhomin}\right\}.
\end{equation}
\end{theorem}

\subsection{Competitive analysis against the offline optimum}

We now compare the residual allocator not against weaker non-crediting baselines, but against the strongest possible comparator: an offline optimum that knows the full workload in advance.

\begin{definition}[Offline optimum and cost measure]
\label{def:offline}
The offline optimum \(\mathrm{OPT}\) knows the full non-adaptive workload \(\mathcal{W}=(q_1,\ldots,q_K)\) in advance. It can plan support allocations to exploit future repeats, and requires no diagnostic budget. The cost measure is total new answering expenditure:
\(\Cans(\pi,\mathcal{W})=\sum_{t=1}^{K}\rho_{t,\pi}^{\mathrm{new}}\).
\end{definition}

\begin{theorem}[1-competitiveness of the residual allocator]
\label{thm:competitive}
Consider the scalar Gaussian exact-reuse setting with additive support. Let \(\mathcal{W}=(q_1,\ldots,q_K)\) be a non-adaptive workload in which all thresholds are feasible for the residual allocator \(\pi^\star\). Then
\[
\Cans(\pi^\star,\mathcal{W})=\Cans(\mathrm{OPT},\mathcal{W}).
\]
The residual allocator is 1-competitive in answering expenditure: the online algorithm matches the offline optimum despite having no knowledge of future queries.
\end{theorem}
\begin{proof}
Partition the positions \(\{1,\ldots,K\}\) into equivalence classes by query identity. For each distinct query \(q\) appearing in \(\mathcal{W}\), let \(T(q)=\{t:q_t=q\}\) and define
\(\rho^{\max}(q)=\max_{t\in T(q)}\rho_t^{\min}\).

\emph{Offline lower bound.} Any allocator must accumulate total support at least \(\rho^{\max}(q)\) for each distinct query \(q\): the occurrence with the largest \(\rho_t^{\min}\) demands at least that much effective support. Conversely, \(\mathrm{OPT}\) achieves this by front-loading \(\rho^{\max}(q)\) at the first occurrence. Therefore
\(\Cans(\mathrm{OPT},\mathcal{W})=\sum_{q\;\mathrm{distinct}}\rho^{\max}(q)\).

\emph{Online residual allocator.} Fix a distinct query \(q\) with ordered occurrences \(t_1<\cdots<t_n\). Let \(A_k=\sum_{j=1}^{k}\rho_{t_j,\pi^\star}^{\mathrm{new}}\) denote accumulated support, with \(A_0=0\). By the residual rule, \(\rho_{t_k,\pi^\star}^{\mathrm{new}}=\max\{0,\rho_{t_k}^{\min}-A_{k-1}\}\).

We show \(A_k=\max_{1\leq j\leq k}\rho_{t_j}^{\min}\) by induction. Base: \(A_1=\rho_{t_1}^{\min}\). Inductive step: if \(\rho_{t_k}^{\min}\leq A_{k-1}\), then \(\rho_{t_k,\pi^\star}^{\mathrm{new}}=0\) and \(A_k=A_{k-1}=\max_{j\leq k}\rho_{t_j}^{\min}\). If \(\rho_{t_k}^{\min}>A_{k-1}\), then \(A_k=A_{k-1}+(\rho_{t_k}^{\min}-A_{k-1})=\rho_{t_k}^{\min}=\max_{j\leq k}\rho_{t_j}^{\min}\). In both cases \(A_k=\max_{j\leq k}\rho_{t_j}^{\min}\).

Therefore \(A_n=\rho^{\max}(q)\), and summing over all distinct queries gives \(\Cans(\pi^\star,\mathcal{W})=\sum_{q\;\mathrm{distinct}}\rho^{\max}(q)=\Cans(\mathrm{OPT},\mathcal{W})\).
\end{proof}

\begin{remark}[Structural role of additive support]
\label{rem:additive}
The 1-competitiveness result is a structural consequence of support additivity (equation (4)). Because \(\rho_t^{\mathrm{eff}}=\rho_t^{\mathrm{hist}}+\rho_t^{\mathrm{new}}\), the residual allocator incrementally accumulates support and reaches the same total as an offline planner that front-loads in a single release. The inductive proof shows that accumulated support tracks the running maximum of threshold requirements---the same quantity the offline planner targets. For non-additive support models (e.g., correlated workloads with sublinear fusion), the competitive ratio could exceed 1.
\end{remark}

\begin{proposition}[Diagnostic overhead]
\label{prop:diagnostic-overhead}
Let \(\xi=\rho_{\mathrm{est}}/\rho_{\mathrm{tot}}\) be the diagnostic budget fraction. In the feasible regime, the full mechanism's total privacy expenditure satisfies
\[
\frac{\rho_{\mathrm{est}}+\Cans(\pi^\star,\mathcal{W})}{\Cans(\mathrm{OPT},\mathcal{W})}
=
1+\frac{\rho_{\mathrm{est}}}{\Cans(\mathrm{OPT},\mathcal{W})}.
\]
If the workload consumes the full answering budget, this ratio is at most \(1/(1-\xi)\). The diagnostic budget fraction \(\xi\) is the sole source of overhead against the offline optimum.
\end{proposition}

\subsection{Scarcity impossibility}
In the scarcity regime, the performance measure shifts to threshold satisfaction: \(\Nsat(\pi,\mathcal{W})=|\{t:\rho_t^{\mathrm{eff}}\geq\rho_t^{\min}\}|\). Here, foresight genuinely helps.

\begin{theorem}[Scarcity impossibility for deterministic online allocators]
\label{thm:scarcity-impossibility}
For any deterministic online allocator \(\pi\) (with or without history crediting) and any integer \(n\geq 2\), there exist two non-adaptive workloads \(\mathcal{W}_A,\mathcal{W}_B\) of distinct queries under budget \(\rho_{\mathrm{ans}}\) such that
\[
\min\!\left\{
\frac{\Nsat(\pi,\mathcal{W}_A)}{\Nsat(\mathrm{OPT},\mathcal{W}_A)},\;
\frac{\Nsat(\pi,\mathcal{W}_B)}{\Nsat(\mathrm{OPT},\mathcal{W}_B)}
\right\}
\leq\frac{1}{n}.
\]
No deterministic online allocator can achieve worst-case competitive ratio greater than \(1/n\) in threshold satisfaction.
\end{theorem}
\begin{proof}[Proof sketch]
Set \(\rho_{\mathrm{ans}}=1\). \(\mathcal{W}_A\): one query with \(\rho^{\min}=1\). \(\mathcal{W}_B\): the same first query with \(\rho^{\min}=1\), followed by \(n\) distinct queries each with \(\rho^{\min}=1/n\). All queries are distinct, so no reuse is possible. If \(\pi\) answers \(q_1\), it exhausts the budget and satisfies \(1\) query on \(\mathcal{W}_B\), while \(\mathrm{OPT}\) skips \(q_1\) and satisfies \(n\), giving ratio \(1/n\). If \(\pi\) skips \(q_1\), it satisfies \(0\) on \(\mathcal{W}_A\) while \(\mathrm{OPT}\) satisfies \(1\), giving ratio \(0\). See Appendix~\ref{app:competitive-proofs} for the full construction.
\end{proof}

\begin{remark}[Phase transition and implications for QIF]
\label{rem:phase-transition}
Theorems~\ref{thm:competitive} and~\ref{thm:scarcity-impossibility} delineate a sharp phase transition. In the feasible regime, online allocation is free: the residual allocator matches the offline optimum (competitive ratio 1). In the scarcity regime, the problem retains the hardness of online selection (competitive ratio \(1/n\)). The structural reason is that feasibility eliminates the combinatorial ``which queries to serve'' decision. This contextualises the QIF scarcity layer: since no online allocator can guarantee a good competitive ratio under worst-case scarcity, the design question shifts from optimality to robustness. QIF provides one principled heuristic for matching limited support to query difficulty, but Theorem~\ref{thm:scarcity-impossibility} shows that no weight function can overcome the fundamental information asymmetry between online and offline allocators.
\end{remark}

\begin{definition}[Reuse-only Scarcity Comparator]
The reuse-only comparator credits $\rhohist$ and uses the same residual rule in the feasible regime, but sets $w_t\equiv1$ in scarcity. Comparing~\eqref{eq:scarcity} with QIF weights against this comparator isolates the contribution of the difficulty weighting rather than transcript reuse.
\end{definition}

\paragraph{Two-query illustration.}
Suppose query \(q_1\) is answered with support \(0.7\), and a later query
\(q_2\) exactly matches \(q_1\). If the threshold for \(q_2\) requires
\(\rho_2^{\min}=1.0\), then the released history already provides
\(\rho_2^{\mathrm{hist}}=0.7\). The residual support required for \(q_2\) is
\(
\rho_2^{\mathrm{new,min}}
=
\max\{0,1.0-0.7\}
=
0.3.
\)
Thus the mechanism spends only \(0.3\) new answering support rather than pricing
\(q_2\) from scratch. An APEx-style allocator (Definition~\ref{def:ac}) would spend \(1.0\); a budget-partitioning allocator (Definition~\ref{def:bp}) would spend \(1.0\) from analyst \(a_2\)'s sub-budget if \(q_1\) was submitted by a different analyst. The offline optimum (Definition~\ref{def:offline}) would also spend a total of \(1.0\) across both queries---matching \(\pi^\star\)'s total of \(0.7+0.3=1.0\), consistent with Theorem~\ref{thm:competitive}.

\paragraph{LLM disclosure.}
All technical content, theorem statements, assumptions, and proofs are the responsibility of the authors. Large language models and editing tools, including ChatGPT and Grammarly, were used to assist with text editing.

\appendix
\section*{Appendix}
\section{QueryModel Interface and Scope}
A QueryModel specifies: a sensitivity bound $\Delta$, a monotone error surrogate $r_t(\rho)$, a diagnostic procedure producing $m_t^{\lb}$ and $\rhomin$, and a fusion rule combining historical and new releases. The formal support-additivity theorem here is for scalar Gaussian exact reuse. Correlated histograms, range queries, or general linear workloads require a covariance-, Gram-, or workload-specific support object before residual subtraction is valid.

\begin{proposition}[Scalar Gaussian support additivity]
\label{prop:support-additivity}
Suppose $\tilde q_t^{\mathrm{hist}}$ and $\tilde q_t^{\mathrm{new}}$ are independent unbiased Gaussian estimates of the same scalar query, with positive supports $\rhohist>0$ and $\rhonew>0$, and variances $1/(2\rhohist)$ and $1/(2\rhonew)$. Then inverse-variance fusion gives
\begin{equation}
\tilde q_t^{\mathrm{fuse}}=\frac{\rhohist\tilde q_t^{\mathrm{hist}}+\rhonew\tilde q_t^{\mathrm{new}}}{\rhohist+\rhonew},
\qquad
\operatorname{Var}(\tilde q_t^{\mathrm{fuse}})=\frac{1}{2(\rhohist+\rhonew)}.
\end{equation}
Thus historical and new support are additive in the scalar Gaussian exact-reuse case.
\end{proposition}

\subsection{Scalar Count Instantiation}
For unit-$\ell_2$-sensitivity scalar counts,
\begin{equation}
\tilde q_t=q_t(D)+\mathcal N\left(0,\frac{1}{2\rho_t}\right).
\end{equation}

Fusion refers to the post-processing step that combines reusable historical DP outputs with any newly released DP answer for the current query. In the scalar Gaussian exact-reuse case, if the historical and new answers estimate the same scalar query, inverse-variance fusion gives
\begin{equation}
\tilde q_t^{\mathrm{fuse}}
=
\frac{
\rho_t^{\mathrm{hist}}\tilde q_t^{\mathrm{hist}}
+
\rho_t^{\mathrm{new}}\tilde q_t^{\mathrm{new}}
}{
\rho_t^{\mathrm{hist}}+\rho_t^{\mathrm{new}}
},
\qquad
\operatorname{Var}(\tilde q_t^{\mathrm{fuse}})
=
\frac{1}{2(\rho_t^{\mathrm{hist}}+\rho_t^{\mathrm{new}})}.
\end{equation}
Thus fusion converts historical and new releases into query-level effective support, while incurring no additional privacy loss beyond the already accounted releases.

\subsection{Disjoint Histogram Bins as Scalar Counts}
For disjoint histograms, each bin can be treated componentwise as a scalar count with its own diagnostic, support surrogate, and fusion rule. This is not a correlated-histogram QueryModel. Correlated histograms require a non-scalar support object.

\section{Notation}
\label{app:notation}

\begin{center}
\small
\begin{tabular}{p{0.24\linewidth}p{0.70\linewidth}}
\hline
\textbf{Notation} & \textbf{Meaning} \\
\hline
\(D\) & Sensitive dataset held by the trusted curator. \\

\(q_t\) & Query arriving at time \(t\). \\

\(K\) & Finite analysis horizon, i.e., the number of query arrivals considered in the composition and comparison results. It is not a privacy parameter. \\

\(U_t\) & Declared accuracy threshold for query \(q_t\). \\

\(\gamma_t\) & Failure probability associated with the accuracy/confidence requirement for query \(q_t\). \\

\(1-\gamma_t\) & Confidence level for the accuracy requirement. \\

\(\rho_{\mathrm{tot}}\) & Global zCDP contract for the full mechanism. \\

\(\rho_{\mathrm{est}}\) & Portion of the zCDP contract allocated to private diagnostics. \\

\(\rho_{\mathrm{ans}}\) & Portion of the zCDP contract allocated to new query execution. \\

\(\rho_{\mathrm{probe},t}\) & zCDP cost of the diagnostic/probe release at time \(t\). \\

\(H_{t-1}\) & Public transcript of previous DP releases before query \(q_t\) arrives. \\

\(\tilde q_t\) & Noisy DP answer released for query \(q_t\). \\

\(\rho_t^{\mathrm{hist}}\) & Historical support already available for query \(q_t\) from previous DP releases. This support has already been paid for in the past. \\

\(\rho_t^{\mathrm{new}}\) & New answering support spent at time \(t\). This is the newly charged privacy cost for answering \(q_t\). \\

\(\rho_t^{\mathrm{eff}}\) & Query-level effective support used for accuracy after combining historical and new support. In scalar Gaussian exact reuse, \(\rho_t^{\mathrm{eff}}=\rho_t^{\mathrm{hist}}+\rho_t^{\mathrm{new}}\). \\

\(\rho_t^{\min}\) & Minimum effective support required to satisfy the declared threshold \(U_t\) from scratch. \\

\(\rho_t^{\mathrm{new,min}}\) & Minimum new support still required after crediting historical support:
\(\rho_t^{\mathrm{new,min}}=\max\{0,\rho_t^{\min}-\rho_t^{\mathrm{hist}}\}\). \\

\(\rho_t^{\mathrm{avail}}\) & New answering support available to allocate at time \(t\). \\

\(r_t(\rho)\) & Error surrogate mapping effective support \(\rho\) to a certified error bound for query \(q_t\). Larger support corresponds to lower error. \\

\(a_t\) & Residual support gap in scarcity:
\(a_t=\rho_t^{\min}-\rho_t^{\mathrm{hist}}\). \\

\(\mu_t\) & Conservation parameter in the scarcity objective, penalising current expenditure under scarcity. \\

\(w_t\) & Positive difficulty weight used in the weighted scarcity objective. \\

\(\operatorname{QIF}_t\) & Query Influence Factor: a diagnostic difficulty and instability signal used to define \(w_t\). It is not an importance ranking over queries. \\

\(m_t^{\mathrm{LB}}\) & Private lower-confidence scale surrogate used in diagnostic calibration and QIF. \\

\(\tilde v_t\) & Diagnostic variability or instability signal used in QIF. It is not the Gaussian noise variance of the released answer. \\

\(b_t\) & Rarity or imbalance signal used in QIF. \\

\(\alpha\) & Exponent controlling the sharpness of the QIF-derived difficulty weight:
\(w_t=\operatorname{QIF}_t^\alpha\). \\

\(\beta\) & Parameter controlling the relative contribution of the rarity/imbalance signal \(b_t\) in QIF. \\

\(g(\cdot)\) & Positive increasing scale normaliser used in QIF. \\

\(\pi^\star\) & Proposed residual allocator. \\

\(\pi^{\mathrm{SA}}\) & State-agnostic threshold-pricing allocator that does not credit historical support. \\

\(\pi^{\mathrm{AC}}\) & Accuracy-aware non-crediting allocator (APEx-style; Definition~\ref{def:ac}). \\

\(\pi^{\mathrm{BP}}\) & Budget-partitioning allocator (Pujol et al.-style; Definition~\ref{def:bp}). \\

\(\mathrm{OPT}\) & Offline optimum that knows the full workload in advance (Definition~\ref{def:offline}). \\

\(\Cans(\pi,\mathcal{W})\) & Total new answering expenditure: \(\sum_t\rho_{t,\pi}^{\mathrm{new}}\). \\

\(\Nsat(\pi,\mathcal{W})\) & Number of threshold-satisfied queries: \(|\{t:\rho_t^{\mathrm{eff}}\geq\rho_t^{\min}\}|\). \\

\(\alpha_{\mathrm{ov}}\) & Empirical exact-overlap fraction in the witness construction. \\

\(\bar{\rho}^{\min}\) & Average target-implied support in the comparison theorem. \\

\(\sigma_{\mathrm{ov}}\) & Budget-weighted reuse savings fraction used in the diagnostic-amortisation identity. \\

\(\xi\) & Diagnostic budget share:
\(\xi=\rho_{\mathrm{est}}/\rho_{\mathrm{tot}}\). \\
\hline
\end{tabular}
\end{center}
\section{Diagnostic Calibration}
\label{app:diagnostics}
For scalar counts, let
\begin{equation}
\tilde m_t=m_t+Z_t,
\qquad Z_t\sim\mathcal N(0,\sigma_{m,t}^2).
\end{equation}
Define
\begin{equation}
m_t^{\lb}=\max\{\tilde m_t-z_{1-\gamma_t}\sigma_{m,t},m_{\min}\},
\qquad
\rhomin=\left(\frac{\kappa_t}{m_t^{\lb}U_t}\right)^2,
\quad \kappa_t=\frac{z_{1-\gamma_t/2}}{\sqrt2}.
\end{equation}
Then $\Pr[m_t\geq m_t^{\lb}]\geq1-\gamma_t$, and the calibrated support does not underestimate the oracle support with probability at least $1-\gamma_t$ under the stated diagnostic model.

\begin{lemma}[Diagnostic inflation bound]
Let $\rho_t^{\mathrm{oracle}}=(\kappa_t/(m_tU_t))^2$. Assume $m_t>(c+z_{1-\gamma_t})\sigma_{m,t}$ and $m_t^{\lb}>m_{\min}$ on the event $Z_t\geq-c\sigma_{m,t}$. Then, on that event,
\begin{equation}
\frac{\rhomin}{\rho_t^{\mathrm{oracle}}}
\leq
\left(\frac{m_t}{m_t-(c+z_{1-\gamma_t})\sigma_{m,t}}\right)^2.
\end{equation}
The event has probability $\Phi(c)$.
\end{lemma}

\section{Approximate Reuse and Bias Envelope}
\label{app:bias}

For approximate reuse, suppose a historical answer targets \(q_{\mathrm{hist}}\) rather than \(q_t\), with
\begin{equation}
|q_t(D)-q_{\mathrm{hist}}(D)|\leq B_t.
\end{equation}
The bias envelope \(B_t\) is supplied externally or derived from public-domain constraints. It affects utility, not privacy: the mechanism either uses \(B_t\) as a guard over already released outputs or falls back to a newly accounted answer.

Let \(h_t\in[0,\rho_t^{\mathrm{hist}}]\) denote the amount of historical support admitted into approximate fusion. This quantity is not the QIF difficulty weight \(w_t\): \(w_t\) controls scarcity allocation, whereas \(h_t\) controls how much historical support is allowed into the approximate fusion step. With fusion weight
\begin{equation}
\lambda_t=\frac{h_t}{h_t+\rho_t^{\mathrm{new}}},
\end{equation}
the squared bias is at most \(\lambda_t^2B_t^2\), and the variance is
\begin{equation}
\frac{1}{2(h_t+\rho_t^{\mathrm{new}})}.
\end{equation}
Hence
\begin{equation}
\operatorname{MSE}(\tilde q_t)
\leq
\left(
\frac{h_t}{h_t+\rho_t^{\mathrm{new}}}
\right)^2B_t^2
+
\frac{1}{2(h_t+\rho_t^{\mathrm{new}})}.
\end{equation}

\section{Proof Details}
\label{app:proofs}
\subsection{QIF-weighted scarcity differentiation proof}
\begin{proof}
For part (1), since \(g(m_i^{\mathrm{LB}})>0\), the numerator
\(\tilde v_i+\beta b_i\) is non-decreasing in both \(\tilde v_i\) and \(b_i\).
Because \(g\) is increasing, increasing \(m_i^{\mathrm{LB}}\) increases the
denominator and therefore weakly decreases \(\operatorname{QIF}_i\). Since
\(\alpha>0\), the transformation \(w_i=\operatorname{QIF}_i^\alpha\) preserves
these monotonicities.

For part (2), in the active-scarcity regime the hinge in the shortfall term is
active, so the objective for query \(i\) is
\begin{equation}
\label{eq:qif-proof-objective}
F_i(\rho_i)=w_i(a_i-\rho_i)^2+\mu\rho_i.
\end{equation}
This is a strictly convex quadratic because \(w_i>0\). Differentiating gives
\begin{equation}
\label{eq:qif-proof-derivative}
F_i'(\rho_i)
=
-2w_i(a_i-\rho_i)+\mu.
\end{equation}
Setting \(F_i'(\rho_i)=0\) yields
\begin{equation}
\label{eq:qif-proof-solution}
\rho_i^\star
=
a_i-\frac{\mu}{2w_i}.
\end{equation}
By assumption, projection does not clip the solution, so this is the constrained
optimiser. The residual shortfall is therefore
\begin{equation}
\label{eq:qif-proof-shortfall}
S_i^\star
=
a_i-\rho_i^\star
=
a_i-\left(a_i-\frac{\mu}{2w_i}\right)
=
\frac{\mu}{2w_i}.
\end{equation}

For part (3), if \(a_s=a_t\), then using~\eqref{eq:qif-compact-result},
\begin{equation}
\label{eq:qif-proof-difference}
\rho_t^\star-\rho_s^\star
=
\left(a_t-\frac{\mu}{2w_t}\right)
-
\left(a_s-\frac{\mu}{2w_s}\right)
=
\frac{\mu}{2}
\left(
\frac{1}{w_s}-\frac{1}{w_t}
\right).
\end{equation}
If \(w_t>w_s>0\), then \(1/w_s>1/w_t\), so this difference is strictly positive.
The shortfall comparison follows directly from~\eqref{eq:qif-compact-result}.
\end{proof}

\begin{remark}[Role and limitation of QIF]
The theorem shows how QIF affects allocation once the weighted scarcity objective
has been chosen. It does not claim that QIF is uniquely optimal, Bayes-optimal,
or regret-minimising. Other positive weights can be substituted into the same
scarcity objective without changing the accounting or closed-form optimisation
results.
\end{remark}

\begin{remark}[Diagnostic variability versus Gaussian noise variance]
The variability signal \(\tilde v_i\) in QIF is not the Gaussian noise variance
of the released answer. It is a diagnostic difficulty signal derived before
allocation and used only to form the scarcity weight. By contrast, the variance
appearing in the scalar Gaussian fusion rule is the estimator/noise variance
induced by the zCDP answering cost, equal to \(1/(2\rho)\).
\end{remark}

\subsection{Memoryless comparison}
A state-agnostic threshold-pricing allocator cannot credit $\rhohist$, so it must allocate from-scratch support when feasible. The residual allocator allocates $\max\{0,\rhomin-\rhohist\}$ in the feasible regime. If $\rhohist>0$ this is strictly smaller than the from-scratch allocation, including the case $\rhohist\geq\rhomin$ where the residual allocation is zero. Dividing by the from-scratch support gives~\eqref{eq:ratio}.

\subsection{Exact-repeat Witness Separation}
Choose a query $q^\dagger$ with target-implied support $\rho^{\min}$ and release it once. Place exact repeats of $q^\dagger$ at $m=\lfloor\alpha_{\mathrm{ov}}(K-1)\rfloor$ later positions, and fill remaining positions with disjoint queries. At each repeat, the residual allocator reuses the earlier DP output, so $\rho_t^{\mathrm{hist}}\geq \rho^{\min}$, and spends zero new answering support. A state-agnostic threshold-pricing allocator allocates from scratch. Summing $m$ gaps gives $m\rho^{\min}$, which is $\Omega(\alpha_{\mathrm{ov}}K\bar\rho^{\min})$ under the constructed homogeneous support case.

\subsection{Diagnostic Amortisation Identity}
Let $\sigma_{\mathrm{ov}}$ be the budget-weighted reuse savings fraction relative to $\pi^{\mathrm{SA}}$, and let $\xi=\rho_{\mathrm{est}}/\rho_{\mathrm{tot}}$. The diagnostic cost is offset whenever
\begin{equation}
\sigma_{\mathrm{ov}}>
\frac{\xi\rho_{\mathrm{tot}}}{\sum_{t=1}^{K}\rho^{\mathrm{new}}_{t,\pi^{\mathrm{SA}}}}.
\end{equation}
This reports how much realised reuse is needed to offset diagnostic expenditure; it does not predict whether a workload will achieve that reuse level.

\section{Competitive Analysis Proofs}
\label{app:competitive-proofs}

\subsection{Proof of Proposition~\ref{prop:diagnostic-overhead}}
\begin{proof}
By Theorem~\ref{thm:competitive}, \(\Cans(\pi^\star,\mathcal{W})=\Cans(\mathrm{OPT},\mathcal{W})\) in the feasible regime. The full mechanism's total expenditure is \(\rho_{\mathrm{est}}+\Cans(\pi^\star,\mathcal{W})\). Dividing by \(\Cans(\mathrm{OPT},\mathcal{W})\) gives the ratio \(1+\rho_{\mathrm{est}}/\Cans(\mathrm{OPT},\mathcal{W})\). When \(\Cans(\mathrm{OPT},\mathcal{W})\geq\rho_{\mathrm{ans}}=(1-\xi)\rho_{\mathrm{tot}}\), the overhead ratio is bounded by \(1+\xi\rho_{\mathrm{tot}}/((1-\xi)\rho_{\mathrm{tot}})=1/(1-\xi)\).
\end{proof}

\subsection{Full proof of Theorem~\ref{thm:scarcity-impossibility}}
\begin{proof}
Fix \(\rho_{\mathrm{ans}}=1\) and an integer \(n\geq 2\). Construct two workloads of mutually distinct queries (no reuse is possible on either):
\begin{itemize}
\item \(\mathcal{W}_A\): A single query \(q_1\) with \(\rho_1^{\min}=1\).
\item \(\mathcal{W}_B\): The same first query \(q_1\) with \(\rho_1^{\min}=1\), followed by \(n\) distinct queries \(q_2,\ldots,q_{n+1}\), each with \(\rho_t^{\min}=1/n\).
\end{itemize}
At \(t=1\), the online allocator \(\pi\) observes \(q_1\) with \(\rho_1^{\min}=1\) and must choose \(\rho_1^{\mathrm{new}}\) without knowing which workload will materialise.

\textbf{Case 1:} \(\pi\) answers \(q_1\) (allocates \(\rho_1^{\mathrm{new}}\geq 1\)). Remaining budget: at most \(0\).
On \(\mathcal{W}_A\): \(\Nsat(\pi,\mathcal{W}_A)=1\) and \(\Nsat(\mathrm{OPT},\mathcal{W}_A)=1\). Ratio: \(1\).
On \(\mathcal{W}_B\): \(\pi\) satisfies only \(q_1\), so \(\Nsat(\pi,\mathcal{W}_B)=1\). \(\mathrm{OPT}\), knowing \(\mathcal{W}_B\), skips \(q_1\) and answers all of \(q_2,\ldots,q_{n+1}\) at total cost \(n\cdot(1/n)=1\). Thus \(\Nsat(\mathrm{OPT},\mathcal{W}_B)=n\), and the ratio is \(1/n\).

\textbf{Case 2:} \(\pi\) does not answer \(q_1\) (allocates \(\rho_1^{\mathrm{new}}<1\)).
On \(\mathcal{W}_A\): \(\Nsat(\pi,\mathcal{W}_A)=0\) and \(\Nsat(\mathrm{OPT},\mathcal{W}_A)=1\). Ratio: \(0<1/n\).

For any deterministic \(\pi\), the adversary selects \(\mathcal{W}_B\) if \(\pi\) answers \(q_1\) (ratio \(1/n\)) and \(\mathcal{W}_A\) if \(\pi\) skips (ratio \(0\)). The best deterministic strategy is to answer \(q_1\), achieving worst-case ratio \(1/n\).
\end{proof}

\section{Implementation Checks}
The implementation checks are supplementary and are not used as formal proof. They test whether the code follows the regimes and inequalities stated in the main paper. In the larger experimental package, the per-query savings bound was checked on 200,000 individual releases without violation, and the exact-repeat avoidable-expenditure construction matched the predicted values across 100 runs. Regime-frequency checks showed that scarcity declined as total budget increased, while cap-bound cases remained rare. Diagnostic checks followed the predicted SNR pattern: over-allocation approached one at high SNR and became large in the lowest-SNR bucket. QIF checks showed meaningful allocation differentiation in non-degenerate medium- and high-overlap workloads, but weak differentiation in the low-overlap control. These checks support the implementation, but do not establish additional theoretical claims.

\end{document}